\documentclass[a4paper,11pt]{article}
\usepackage[english]{babel}
\usepackage[utf8]{inputenc}
\usepackage{paralist,url,verbatim,anysize}
\usepackage{amscd}
\usepackage{fancyhdr}
\usepackage{mathrsfs}
\usepackage{nicefrac}
\usepackage[e]{esvect}
\usepackage{amsmath,amsthm,amssymb,amsfonts}
\usepackage{graphicx, graphics}
\usepackage[bookmarksopen=true]{hyperref}
\hypersetup{colorlinks=true, citecolor=blue}
\usepackage[usenames, dvipsnames]{color}
\usepackage{bbm}
\usepackage[font=small,labelfont=bf]{caption}
\usepackage{subcaption}

\theoremstyle{plain}
\newtheorem{theorem}{\bf Theorem}[]

\newtheorem{corollary}[theorem]{Corollary}
\newtheorem{lemma}[theorem]{Lemma}
\newtheorem{proposition}[theorem]{Proposition}
\newtheorem{example}[theorem]{Example}

\newtheorem{thmnonumber}{\bf Main Theorem}

\theoremstyle{definition}
\newtheorem{remark}[theorem]{Remark}
\newtheorem{definition}[theorem]{Definition}

\newcommand{\link}{\mathrm{link}\, }
\renewcommand{\star}{\mathrm{star}\, }

\newcommand{\cm}[1]{}

\definecolor{mypink}{RGB}{215, 5, 234}
\newcommand{\mypink}{\color{mypink} \bf}
\newcommand{\spread}{\rm \textsc{spread}}
\newcommand{\fold}{\rm \textsc{fold}}
\renewcommand{\split}{\textsc{split}}
\newcommand{\unite}{\textsc{unite}}
\newcommand{\longhookrightarrow}{\ensuremath{\lhook\joinrel\relbar\joinrel\rightarrow}}
\begin{document}

\author{
Bruno Benedetti \thanks{Supported by NSF Grant  1600741, ``Geometric Combinatorics and Discrete Morse Theory''.} \\
\small Department of Mathematics\\ \small University of Miami\\
\small Coral Gables, FL 33146\\
\small \url{bruno@math.miami.edu}}
\date{\small August 6, 2016}
\title{Mogami manifolds, nuclei, and $3D$ simplicial gravity}
\maketitle


\begin{abstract}
Mogami introduced in 1995 a large class of triangulated $3$-dimensional pseudomanifolds, henceforth called ``Mogami pseudomanifolds''. He proved an exponential bound for the size of this class in terms of the number of tetrahedra. The question of whether all $3$-balls are Mogami has remained open since; a positive answer would imply a much-desired exponential upper bound for the total number of $3$-balls (and $3$-spheres) with $N$ tetrahedra.

Here we provide a negative answer: many $3$-balls are not Mogami. On the way to this result, we characterize the Mogami property in terms of nuclei, in the sense of Collet--Eckmann--Younan: ``The only three-dimensional Mogami nucleus is the tetrahedron''.
\end{abstract}


\section*{Introduction}

A long standing open question in discrete geometry (also highlighted by Gromov, cf.~\cite[pp.~156--157]{GromovQuestion}) is whether there are exponentially many simplicial complexes homeomorphic to the $3$-sphere, 
or more than exponentially many. What is counted here is the number of \emph{combinatorial types}, in terms of the number $N$ of tetrahedra. This enumeration problem is crucial for the convergence of a certain model in discrete quantum gravity, called ``dynamical triangulations''; see for example the book \cite{ADJ} or the survey \cite{ReggeWilliams} for an introduction. 

By deleting one simplex from any (triangulated) $3$-sphere, we obtain a (triangulated) $3$-ball. Conversely, by coning off the boundary of any $3$-ball, we get a $3$-sphere. This close relation between $3$-spheres and $3$-balls is reflected in the asymptotic enumeration. In fact, it is not hard to see that $3$-balls are more than exponentially many if and only if $3$-spheres are. In other words, one can equivalently rephrase our enumeration problem by replacing ``$3$-sphere'' with ``$3$-ball''. 

To tackle the problem, in 1995 Durhuus and J\'{o}nsson introduced the class of {Locally Constructible} ({``LC}'') manifolds, for which they were able to prove an exponential upper bound \cite{DJ} \cite[Theorem 4.4]{BZ}. The geometric idea is ingeniously simple. Let us agree to call {\mypink tree of $d$-simplices} any triangulated $d$-ball whose dual graph is a tree. Definitorially, {\mypink LC manifolds} are those triangulations of manifolds with boundary that can be obtained from some tree of $d$-simplices by repeatedly gluing together two \emph{adjacent} boundary facets. This adjacency condition for the matching, together with the fact that trees are exponentially many, results in a global exponential upper bound.

Durhuus and J\'{o}nsson conjectured that all $3$-spheres (and all $3$-balls) are LC. This was disproven only recently by the author and Ziegler \cite{BZ}. The key for the disproval was a characterization of the LC property in terms of simple homotopy theory: ``A $3$-sphere is LC if and only if it admits a discrete Morse function with exactly two critical faces'' \cite[Cor.~2.11]{BZ}. Knot theory provides then obstructions to the latter property.

\enlargethispage{3mm}
In 1995, Mogami introduced another class of manifolds, henceforth called ``{\mypink Mogami manifolds}'' \cite{Mog}. Essentially, these are the triangulations of manifolds with boundary that can be obtained from a tree of $d$-simplices by repeatedly gluing together two \emph{incident} boundary facets (Remark \ref{rem:Characterization}). Since ``adjacent'' implies ``incident'', LC obviously implies Mogami. The converse is false: here we prove that a cone is Mogami if and only if its basis is strongly-connected (\textbf{Proposition \ref{prp:Cone}}), 
so many cones are Mogami but not LC. 

Building on top of Durhuus--J\'{o}nsson's work, Mogami was able to show an exponential bound also for his broader class of manifolds. Mogami's argument is based on link planarity and is specific to dimension $3$, whereas Durhuus--J\'{o}nsson's argument can be extended to arbitrary dimension \cite[Theorem 4.4]{BZ}. Still, an interesting conjecture arises from Mogami's work: Perhaps all $3$-balls or $3$-spheres are Mogami, even if not all of them are LC  \cite[p.~161]{Mog}. 

Mogami's conjecture is weaker than Durhuus--J\'{o}nsson's, 
but a positive solution would still solve the enumeration problem: It would imply that there are only exponentially many $3$-balls. Mogami's conjecture is harder to tackle, mainly because we lack a characterization of the Mogami property in terms of simple homotopy theory. Hence  the methods that allowed to solve Durhuus-J\'{o}nsson's conjecture do not extend.

Meanwhile, in 2014 Collet, Eckmann and Younan showed that the total number of $3$-spheres or $3$-balls crucially depends on the number of $3$-balls that have all vertices on the boundary.  More specifically: Let us call {\mypink nucleus} a $3$-ball with all vertices on the boundary, and in which every interior triangle contains at least two interior edges. (The notion was first introduced by Hachimori, under the name ``reduced ball'' \cite[p.~85]{HachiThesis}; the name ``nucleus'' appears in \cite{CEY}). The enumeration problem of $3$-balls (or $3$-spheres) is equivalent to the question of whether nuclei are exponentially many, or more  \cite[Theorem 5.17]{CEY}.

In the present paper, we combine Mogami's and Collet--Eckmann-Younan's intuitions, by characterizing the Mogami property among $3$-balls without interior vertices.
 
\begin{thmnonumber}[Corollary \ref{cor:Nuclei} \& Theorem \ref{thm:final}]  The only Mogami nucleus is the tetrahedron. Moreover, for $3$-balls \emph{without interior vertices}, the following inclusions hold: \em
\[
\{  \textrm{shellable}  \} \subsetneq 
\{  \textrm{LC}  \} =
\{  \textrm{Mogami}  \} \subsetneq
\{  \textrm{collapsible}  \} \subsetneq 
\{  \textrm{all $3$-balls without interior vertices}  \}   .
\]
\end{thmnonumber}

In particular, Bing's thickened house with two rooms \cite{HachiLibrary} and all non-trivial nuclei listed in \cite{CEY} yield counterexamples to Mogami's conjecture. Using knot theory, we can even give a coarse estimate for the asymptotic number of non-Mogami balls: 

\begin{thmnonumber}[Lemma \ref{lem:Knot1} \& Propositions \ref{prop:asymptotics}, \ref{prop:2gen}, \ref{prop:2bridge}.] Let $B$ be a $3$-ball with a knotted spanning edge and with all vertices on the boundary. If the knot is 
\begin{compactitem}
\item a single trefoil, then $B$ can be collapsible but it cannot be Mogami;
\item a connected sum of $2$ or more trefoils, then $B$ is neither Mogami nor collapsible.
\end{compactitem}
Moreover, the number of non-Mogami $3$-balls without interior vertices is asymptotically the same as the total number of $3$-balls without interior vertices. 
\end{thmnonumber}

With this,  the problem of enumerating combinatorial types of $3$-balls remains wide open. All the known strategies expected to succeed in showing an exponential bound (cf.~e.g.~\cite[295--296]{ADJ}) have currently failed. A combinatorial criterion that divides the entire family of triangulated $3$-manifolds (or $d$-manifolds, for any fixed $d$) into nested subfamilies, each of exponential size, was introduced in \cite{Benedetti-DMT4MWB};
metric restrictions on triangulations that also give exponential bounds 
have been discovered in \cite{AB-MGCC}.


\section*{Methods}
Our proof is technical but the main idea is elementary, and best sketched with an example. In Figure \ref{fig:VertexStar}, we show a portion of the boundary of some nicely triangulated $3$-ball $B$; specifically, the star of a vertex $v$ in $\partial B$. For brevity throughout the paper we say ``{\mypink boundary-link of $v$}''  instead of ``link of $v$ in the boundary of $B$''.

\begin{figure}[htbp]
\begin{center}
\includegraphics[scale=0.8]{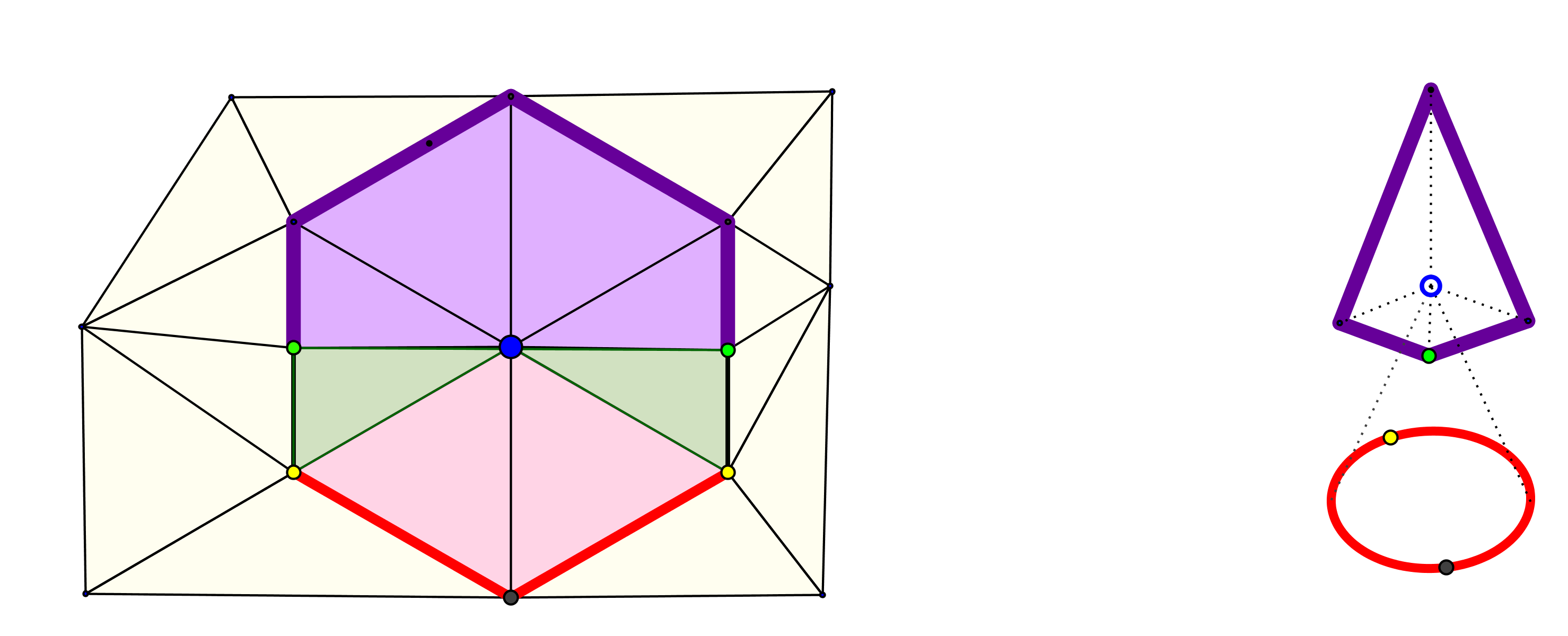} $\qquad \qquad \quad$
\includegraphics[scale=0.6]{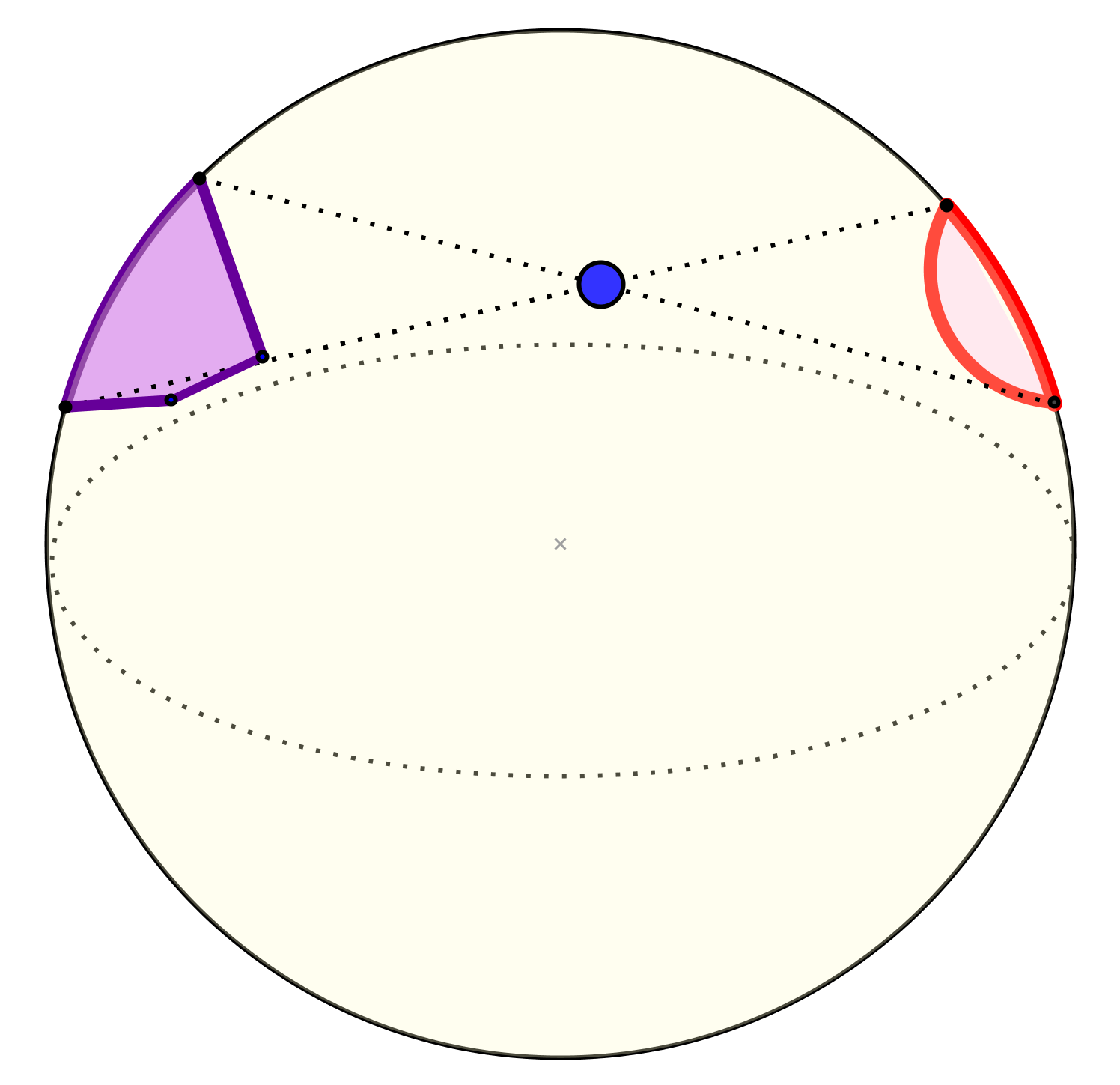}
\caption{\textsc{Left}: Part of the boundary of a $3$-ball $B$. The identification of the green triangles is a Mogami step. Once we perform it, the link inside $\partial B$ of the blue vertex ``splits'' (\textsc{center}): from a circle, it becomes two disjoint circles. Topologically, $\partial B$ gets ``pinched'' at the blue vertex (\textsc{right}).} 
\label{fig:VertexStar}
\end{center}
\end{figure}

\vskip-2mm

The green triangles are incident at $v$, but not adjacent. Their identification is a ``Mogami gluing'', but not an ``LC gluing'' (cf.~Definitions \ref{def:LCgluing}, \ref{def:Mogamigluing}). As depicted in Figure~\ref{fig:VertexStar}, the gluing changes the topology of the boundary: $v$  becomes  a {\mypink singularity}, in the sense that its link is disconnected. Also, after the gluing, we no longer have a simplicial complex, because the pink triangles now share 2 edges out of 3. We call such configuration of two boundary triangles sharing exactly 2 edges a {\mypink wound}.  Let us now perform a second identification, namely, let us glue the pink triangles together. As we do that, the topology changes back: The boundary-link of $v$ returns to be a single circle, as one of its connected components, the red digon, is sunk into the interior. The step of gluing together two boundary triangles with exactly 2 edges in common is called a {\mypink healing}. (The same step was called ``type-(iv) LC gluing'' in \cite[Definition 3.17]{BZ}.) The healing makes the wound disappear, as the triangle resulting from the identification is sunk into the interior. 

Now, let us start back from $B$ and let us perform the same two gluings in \emph{inverse order}: Pink first, then green. There are two pleasant novelties with this reshuffling:  
\begin{compactenum}[(1)]
\item When the pink triangles are glued, they share $1$ edge, not $2$. So the gluing is not a healing. 
\item When the green triangles are glued, they share $1$ edge, not just one vertex. As a result, the ``green gluing'' is now a legitimate LC gluing.
\end{compactenum}
By postponing the Mogami-non-LC move to after the healing move, topologically these two bizarre moves have `canceled out'; and we have obtained a sequence in which all triangles that we match have exactly $1$ edge in common at the moment of the gluing. The final complex is obviously the same ball as before.

Using this idea, we will prove that all Mogami $3$-ball without interior vertices are LC (Theorem~\ref{thm:main}). The trick is to systematically rearrange the Mogami sequence to obtain a sequence that is also LC. This does not work for all pseudomanifolds; but if we focus on Mogami constructions of $3$-balls without interior vertices, we know that the boundary-link of every vertex should eventually become a $1$-sphere. Hence, all the extra components of a boundary-link created by Mogami non-LC gluings have to be suppressed throughout the construction. Now, the only way to suppress a component is via a ``healing'' step. By reshuffling, we will obtain a new sequence where the non-LC step and the healing step `cancel out'.



\section*{Notation}
Throughout this paper,  $d$ is always an integer $\ge 2$.  
For the definitions of simplicial complex, regular CW complex, pure, shellable, cone..., we refer the reader to \cite{BZ}. 
Following \cite{BZ}, by {\mypink pseudomanifold} we mean a finite regular CW complex which is pure $d$-dimensional, simplicial, and such that every $(d-1)$-cell belongs to at most two $d$-cells. The {\mypink boundary} is the smallest subcomplex of the pseudomanifold containing all the $(d-1)$-cells that belong to exactly one $d$-cell. We call ``{\mypink $d$-ball}'' (resp.``{\mypink $d$-sphere}'' )  any simplicial complex homeomorphic to the unit ball in 
$\mathbb{R}^d$ (resp. to the unit sphere in $\mathbb{R}^{d+1}$). A {\mypink tree of $d$-simplices} is any $d$-ball whose dual graph is a tree.

\begin{definition}[\unite; \split]
Let $P_1, P_2$ be two disjoint $d$-pseudomanifolds, $d \ge 2$. The operation  {\mypink \unite} consists in identifying a $(d-1)$-face $\Delta'$ in  $\partial P_1$ with a $(d-1)$-face  $\Delta''$ in  $\partial P_2$. (For $d=3$, this was called ``step of type (i)'' in \cite[Definition 3.17]{BZ}.) If the $P_i$ do not have interior vertices,  neither does the obtained pseudomanifold $Q$; and if both $P_i$'s are $d$-balls, so is $Q$. Note also that $Q$ contains in its interior a $(d-1)$-face $\Delta$ with $\partial \Delta$ completely contained in $\partial Q$. 

The inverse operation is called  ``{\mypink split}''. (For $d=3$, this goes under the name of ``Cut-a-3-face'' in \cite[p.~267]{CEY} and of ``Operation (I)'' in \cite[p.~85]{HachiThesis}.)  It is defined whenever a pseudomanifold $Q$ has some  interior $(d-1)$-face $\Delta$ with  $\partial \Delta \subset \partial Q$. If $Q$ is simply-connected, the effect of {\split}  is to divide $Q$ (along the face $\Delta$) into two disconnected pseudomanifolds. In general, the effect of {\split} on the dual graph of the pseudomanifold is to delete one edge. 
\end{definition}

Trees of $N$ $d$-simplices are characterized as the $d$-complexes obtainable from $N$ disjoint $d$-simplices via exactly $N-1$ {\unite} steps.


\begin{definition}[\fold; \spread] \label{def:FoldSpread}
Let $P$ be a $d$-pseudomanifold, $d \ge 2$.
The operation {\mypink \fold} consists in identifying two boundary facets $\Delta', \Delta''$ that share {\em exactly one} $(d-2)$-face  $e$; compare Figure \ref{fig:spread}. (For $d=3$, the operation was called ``an LC step of type (ii)'' in \cite[Definition 3.17]{BZ}.)  If $P$ is a $d$-ball, then the obtained pseudomanifold $Q$ is homeomorphic to $P$. (This is false if $P$ is an arbitrary pseudomanifold, cf. Example \ref{ex:NonHom}.) Moreover, if $d \ge 3$ and $P$ does not have interior vertices,  neither does $Q$. The obtained pseudomanifold $Q$ contains in its interior a $(d-1)$-face $\Delta$ with exactly $d-1$ of its facets in $\partial Q$: in fact,  the only facet of $\Delta$ in the interior of $Q$ is the $(d-2)$-face $e$.

The inverse operation is called {\mypink \spread}; compare Figure \ref{fig:spread}. 
(For $d=3$, it goes under the name `Open-a-2-face'' in \cite[p.~267]{CEY} and ``Operation (II)'' in \cite[p.~85]{HachiThesis}). It is defined whenever a pseudomanifold $Q$ has some interior $(d-1)$-face $\Delta$ that has one of its $(d-2)$-faces in the interior of $Q$, and all its other $(d-2)$-faces in the boundary of $Q$.  
\end{definition}

\begin{figure}[htbp]
\begin{center}
\includegraphics[scale=0.31]{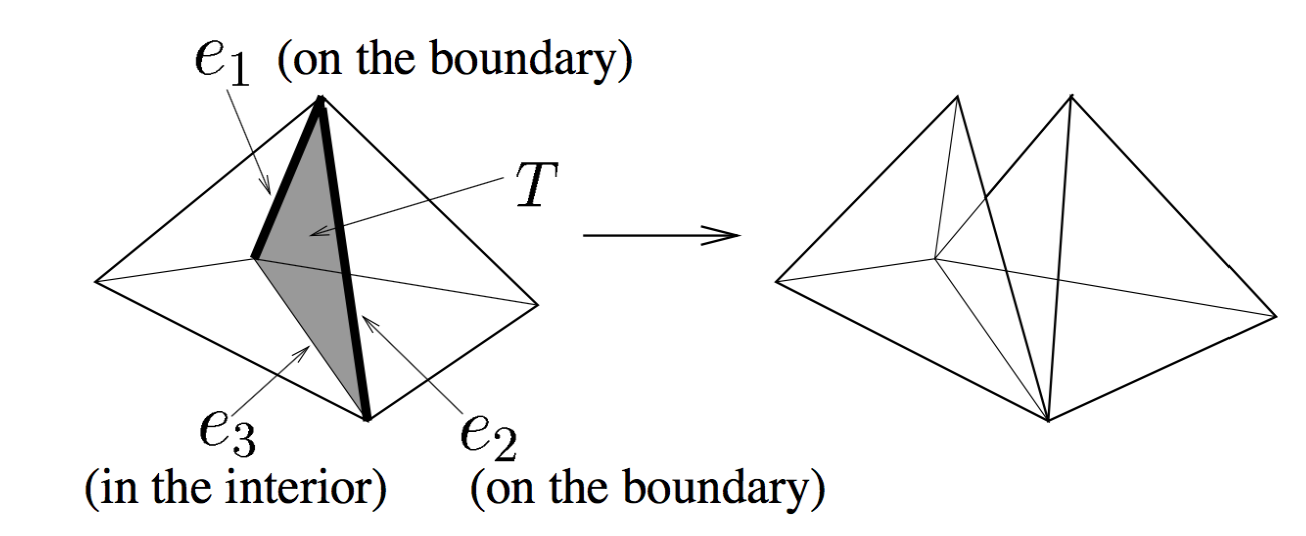} $\qquad \qquad \quad$
\caption{A \textsc{spread} operation for $d=3$. (Picture taken from \cite[p.~85]{HachiThesis}.) The inverse move -- namely, to identify two boundary triangles with exactly one edge in common --- is called \textsc{fold}.} 
\label{fig:spread}
\end{center}
\vskip-8mm
\end{figure}

When {\spread} is applied to a simplicial complex, it outputs a simplicial complex. In contrast, it is easy to see that {\fold} moves may lead out of the world of simplicial complexes.

Next, we introduce nuclei, which were called ``reduced balls'' in \cite[p.~85]{HachiThesis}:
 
\begin{definition}[Nucleus] 
Let $d\ge 2$. A {\mypink nucleus} is a $d$-ball where
\begin{compactenum}[(1)]
\item every $(d-3)$-face belongs to the boundary, and 
\item every interior $(d-1)$-face has at least $d-1$ of its $d$ facets in the interior of the ball. 
\end{compactenum}
 The $d$-simplex (for which condition (2) is void) is called the {\mypink trivial nucleus}.
 \end{definition}

The only $2$-dimensional nucleus is the trivial one. For $d \ge 3$, however, many non-trivial $d$-nuclei exist \cite{CEY}; for example, Hachimori's triangulation of Bing's thickened house with 2 rooms~\cite{HachiLibrary}. 

\begin{lemma}[Hachimori \cite{HachiThesis}, Collet--Eckmann--Younan \cite{CEY}]
Every $3$-ball without interior vertices can be reduced to a disjoint union of nuclei with some (greedy) sequence of {\split} and {\spread} moves. Without loss of generality, one can assume that all {\spread} steps are performed before the {\split} ones. \end{lemma}

The next move can be viewed as a variation/expansion of {\fold}. 

\begin{definition}[LC gluing] \label{def:LCgluing}
Let $P$ be a $d$-pseudomanifold, $d \ge 2$. 
Identifying two boundary facets $\Delta', \Delta''$ whose intersection is \emph{$(d-2)$-dimensional} is an operation called an {\mypink LC gluing}.\end{definition}

Every fold is an LC gluing. The converse is false: for example, when $d=2$, gluing together two boundary edges that have \emph{both} endpoints in common is an LC gluing, but not a fold. The difference is topologically remarkable. It was proven in \cite{BZ} that the only manifolds obtainable from a tree of $d$-simplices with {\fold} moves, are $d$-balls. In contrast, with LC gluings one can obtain all polytopal $d$-spheres, for example. It was proven in \cite{B-Smoothing} that except when $d=4$, all simply-connected smooth $d$-manifolds (with or without boundary!) have a triangulation that can be obtained from some tree of simplices via LC gluings (cf.~Theorem \ref{thm:SimplyConnected}).  

Here is a further generalization, potentially leading to a broader gauge of complexes:

\begin{definition}[Mogami gluing] \label{def:Mogamigluing}
Let $P$ be a $d$-pseudomanifold, $d \ge 2$.
Identifying two boundary facets $\Delta', \Delta''$ whose intersection is \emph{nonempty} is an operation called a {\mypink Mogami gluing}.\end{definition}

Clearly, every LC gluing is a Mogami gluing, while the converse is false (unless $d=2$). We have arrived to the most important definition of the paper:

\begin{definition}[LC manifolds; Mogami manifolds]
Let $d\ge 2$. Let $M$ be a pure $d$-dimensional simplicial complex with $N$ facets that is also a pseudomanifold. $M$ is called 
 {\mypink LC} (resp.~{\mypink Mogami}) if it can be obtained from a tree of $N$ $d$-simplices via some sequence, possibly empty, of LC gluings (resp.~of Mogami gluings). We refer to the sequence as ``the LC construction'' (respectively, ``the Mogami construction''). With abuse of notation, the intermediate pseudomanifolds in the LC construction of an LC manifold are also called ``LC pseudomanifolds''; same for Mogami.
 \end{definition}
 
\begin{remark} \label{rem:Characterization}
The original definition of \cite{Mog}, given only for $d=3$, was slightly different. Mogami considered a class $\mathfrak{C}$ of $3$-pseudomanifolds obtained from a tree of tetrahedra by performing either (1) LC gluings, or (2) identifications of incident boundary edges, subject to a certain planarity condition. 

Now, identifying $2$ boundary edges that share a vertex $v$ creates new adjacencies between triangles that before were only incident at $v$. So it is clear that Mogami $3$-pseudomanifolds (with our definition) all belong to the class $\mathfrak{C}$, since we could realize any Mogami gluing as a ``combo'' of an identification of adjacent boundary edges followed by an LC gluing. 

Conversely, we claim that all manifolds in $\mathfrak{C}$ are Mogami. (This is false for pseudomanifolds.) 
In fact, if we identify two boundary edges that share a vertex $v$ in the boundary of an arbitrary pseudomanifold, we create an entire ``singular edge''. To get a manifold, we have to get rid of this singular edge; the only way to do so is by identifying two triangles $\Delta', \Delta''$ containing that edge, at some point in the Mogami construction. But then we can rearrange the sequence of gluings by performing the Mogami gluing $\Delta' \equiv \Delta''$ before all other gluings.
 \end{remark}

\section{General Aspects of Mogami Complexes}
Let us start with a topological motivation to study the Mogami class. 

\begin{proposition} \label{prop:SimplyConnected} Every Mogami $d$-pseudomanifold is simply-connected.
\end{proposition}

\begin{proof}
By induction on the number of Mogami gluings. Any tree of simplices is topologically a ball, hence simply connected. Consider now the moment in which we glue together two incident boundary facets $\Delta'$ and $\Delta''$ of a simply-connected $d$-pseudomanifold $P$; and suppose a new loop arises. This means that we have just identified two endpoints $x' \in \Delta'$ and $x'' \in \Delta''$ of a path whose relative interior lies completely in the interior of $P$. Let $v$ be a vertex in $\Delta' \cap \Delta''$. By homotoping both $x'$ and $x''$ to $v$, one sees that the ``new loop'' is actually homotopy equivalent to an ``old loop'' already contained in $P$ (hence homotopically trivial, by induction.)
\end{proof}

Not all triangulations of simply-connected manifolds are Mogami, as we will prove in Theorem \ref{thm:final}. However, a partial converse to Proposition \ref{prop:SimplyConnected}  can be derived from \cite{B-Smoothing}:

\begin{theorem}[Benedetti \cite{B-Smoothing}] \label{thm:SimplyConnected} For $d\ne4$, any PL triangulation of any simply-connected $d$-manifold (with boundary) becomes an LC triangulation after performing a suitable number of consecutive barycentric subdivisions.
\\ In particular, every simply-connected smooth $d$-manifold $d \ne 4$, admits a Mogami triangulation. 
\end{theorem}

Recall that a simplicial complex is called {\mypink strongly-connected} if it pure (i.e.~all facets have the same dimension) and its dual graph is connected. By induction on the number of Mogami steps, one can easily prove:


\begin{proposition} \label{prop:StronglyConnected} Every Mogami $d$-pseudomanifold is strongly-connected, and all vertex links in it are strongly-connected as well.
\end{proposition}

The converse does not hold: any triangulation of an annulus is strongly-connected and has strongly-connected links, but it cannot be Mogami by Proposition~\ref{prop:SimplyConnected}.

For $2$-dimensional pseudomanifolds, the LC property and the Mogami property are equivalent, because two boundary edges are adjacent if and only if they are incident. We show next that the two properties diverge from dimension $3$ on. 

In \cite[Lemma 2.23]{BZ} it is shown that {\em the union of two LC pseudomanifolds with a codimension-one strongly-connected intersection, is LC}. Interestingly, an analogous result holds for the Mogami property, basically up to replacing ``strongly-connected'' with ``connected'': 


\begin{proposition} \label{prp:MogamiConstruction}
Let $A, B, C$ be three $d$-pseudomanifolds such that $A \cup B = C$. Assume $A \cap B$ is pure $(d-1)$-dimensional and connected. If $A$ and $B$ are both Mogami, so is~$C$.
\end{proposition}

\begin{proof} First of all, we observe that $A \cap B$ is contained in both $\partial A$ and $\partial B$. In fact, since $A \cup B$ is a pseudomanifold, every $(d-1)$-face of $A \cap B$ can be contained in at most two $d$-faces of $A \cup B$, so it has to be contained in exactly one $d$-face of $A$ and in exactly one $d$-face of $B$.

Since  $A \cap B$ is connected, we can find a total order $F_0, \ldots, F_s$ of the facets of $A \cap B$ such that
for each $i \ge 1$, $F_i$ is incident to some $F_j$, with  $j<i$.
Let us fix  a Mogami construction for $A$ and one for $B$. Let $T_A$ (resp.\ $T_B$) be the tree of $d$-simplices from which $A$ (resp.\ $B$) is obtained. If we perform a  {\unite} move and join $T_A$ and $T_B$ ``at $F_0$'', we obtain a unique tree of tetrahedra $T_C$ containing all facets of $C$. Each $F_i$ ($i \ge 1$) corresponds to two distinct $(d-1)$-faces in the boundary of $T_C$, one belonging to  $T_A$ and one to $T_B$; we will call these two faces ``the two copies of $F_i$''. Now $C$ admits a Mogami construction starting from $T_C$, as follows:
\begin{compactenum}[(a)]
\item first we perform all identifications of boundary facets of $T_C$ that belonged to $T_A$, exactly as prescribed in the chosen Mogami construction of $A$ from $T_A$;
\item then we perform the identifications given by the Mogami construction of $B$;
\item finally, for each $i \ge 1$ (and in the same order!), we glue together the two copies of $F_i$. \end{compactenum}
Since each $F_i$ is incident to some $F_j$, with $j<i$,  the gluings of phase~(c) are Mogami gluings. \qedhere

\end{proof}

\begin{corollary} \label{cor:MogamiNotLC}
Some $3$-dimensional pseudomanifolds are Mogami, but not LC.
\end{corollary}

\begin{proof}
Let $C_1$ and $C_2$ be two shellable simplicial $3$-balls consisting of $4$ tetrahedra, as indicated in Figure \ref{fig:Gluing}. (The $3$-balls are cones over the subdivided squares on their front.) 
\begin{figure}[htbp]
\begin{center}
\includegraphics[scale=0.58]{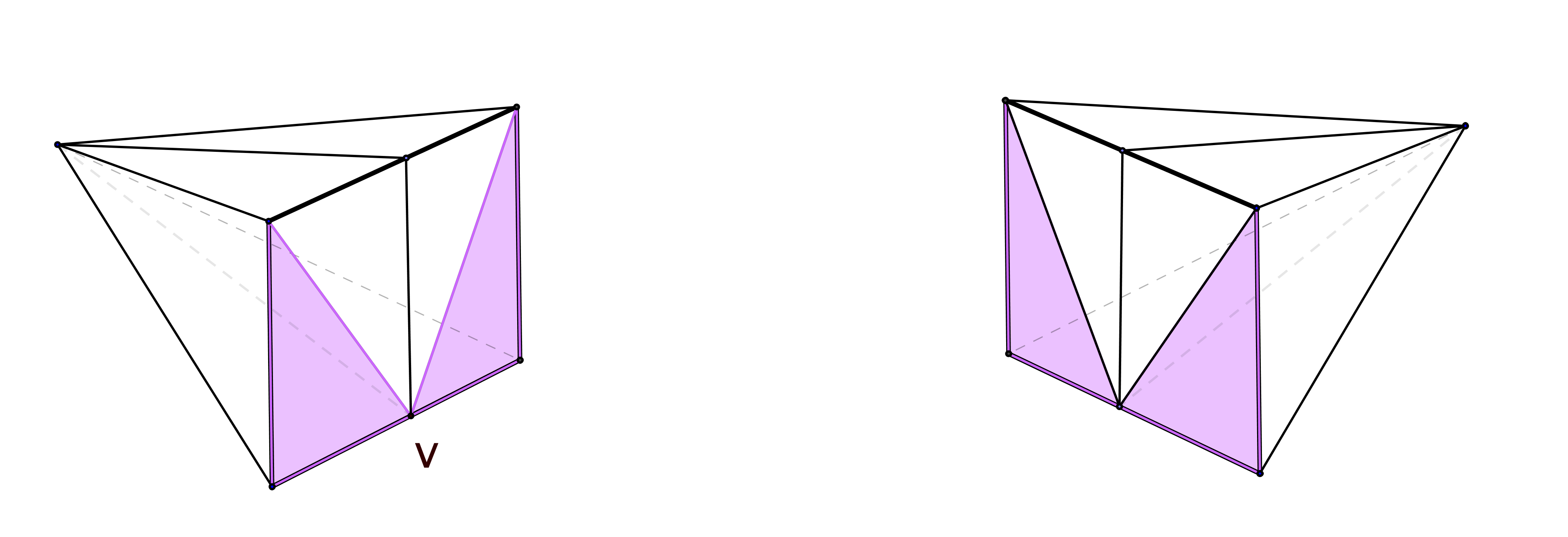}
\caption{Gluing the simplicial $3$-balls along the shaded $2$-dimensional subcomplex (which is connected, but not strongly-connected) gives a Mogami $3$-pseudomanifold that is not LC. Note that the resulting $3$-pseudomanifold is a cone over an annulus (the vertex $v$ is the apex of the cone).}
\label{fig:Gluing}
\end{center} \vskip-1mm
\end{figure}
Since shellable implies LC \cite{BZ} and LC implies Mogami, both $C_1$ and $C_2$ are Mogami. Glue them together in the shaded subcomplex in their boundary  (which uses $5$ vertices and $2$ triangles.) Note that such subcomplex is connected, but not strongly-connected. Let $P$ be the resulting $3$-dimensional pseudomanifold. By Proposition \ref{prp:MogamiConstruction} the pseudomanifold $P$ is Mogami. It remains to prove that $P$ cannot be LC. For this we use a topological result by Durhuus and J\'{o}nsson \cite{DJ}: If $L$ is any LC $3$-dimensional pseudomanifold, then 
any strongly-connected component of $\partial L$ is a $2$-sphere; in addition, any two strongly-connected components of $\partial L$ intersect in at most one point. Yet our $\partial P$ has a different topology: It is a ``pinched sphere'', i.e. the space obtained by identifying two antipodal points of a $2$-sphere. Hence, $P$ cannot be LC. (Alternatively, one can also observe that $P$ is a cone over an annulus; an annulus is not simply connected and therefore not LC; via \cite[Proposition 3.25]{BZ}, this implies that $P$ cannot be LC either.)
\end{proof}

We have arrived to another crucial difference between the LC and the Mogami notion, namely, the behavior with respect to taking cones. In \cite[Proposition 3.25]{BZ} it is proven that {\em for any pseudomanifold $P$ and for any vertex $v$ not in $P$, the cone $v \ast P$ is LC if and only if $P$ is LC}. It turns out that cones tend to be Mogami more often.

\begin{proposition} \label{prp:Cone}
Let $A$ be a $d$-pseudomanifold. Let $v$ be a new point. The cone $v \ast A$ is Mogami if and only if $A$ is strongly-connected. 
\end{proposition}

\begin{proof} The ``only if'' part follows from Proposition \ref{prop:StronglyConnected}, since the link of $v$ in $v\ast A$ is $A$ itself. 
As for the ``if'': Since the dual graph of $A$ is connected, we may choose a spanning tree, which uniquely determines a tree of $d$-simplices $T_N$ inside $A$. Since every $(d-1)$-face of $A$ belongs to at most two $d$-simplices, the complex $A$ can be obtained from $T_N$ via identifications of pairs of (not necessarily incident!) boundary facets. 
Now let us take a new vertex $v$. Clearly $v \ast T_N$ is a tree of $(d+1)$-simplices. Let us `mimic' the construction of $A$ from $T_N$, to obtain a construction of $v \ast A$ from $v \ast T_N$.  (By this we mean that if the construction of $A$ from $T_N$ started by gluing two faces $\sigma'$ and $\sigma''$ of $\partial T_N$, then we should start the new construction of $v \ast A$ from taking $v \ast T_N$ by gluing $v \ast \sigma'$ with $v \ast \sigma''$; and so on.) Clearly, $v \ast A$ is obtained from $v \ast T_N$ via identifications of pairs of boundary facets {\em that contain $v$}, and therefore are incident.
\end{proof}

\begin{corollary}
For each $d \ge 3$, some $d$-dimensional pseudomanifold is Mogami, but not LC.
\end{corollary}

\begin{proof} Let $k$ be any integer such that $2 \le k \le d-1$. Let $A$ be any $k$-pseudomanifold that is strongly-connected, but not LC. (They exists; for example, for $k=2$ one can choose any triangulation of an annulus; compare Figure \ref{fig:Gluing}, which illustrates the case $d=3$.) 
Take $d-k$ consecutive cones over $C$. The resulting $d$-complex is Mogami by Proposition \ref{prp:Cone} and not  LC by \cite[Proposition 3.25]{BZ}.
\end{proof}

\section{Intermezzo: Planar matchings and extensively-LC manifolds}
Here we show that all $2$-spheres and $2$-balls are Mogami and even LC  \emph{independently from which tree of triangles one starts with}. These results are not new; they essentially go back to Durhuus, cf.\ \cite{Durhuus} \cite[p.~184]{DJ}, but we include them to showcase some proof mechanisms that will later be needed in the $3$-dimensional case. We also discuss a higher-dimensional extension of this phenomenon of ``irrelevance of the chosen tree'', called ``extensively-LC'' property. The reader eager for new theorems may skip directly to the next Section.

We need some additional notation. By a {\mypink cycle} we mean from now on a simple cycle; that is, any closed path in which all vertices are distinct, except for the first and last one. A {\mypink graph} (resp.~a {\mypink multigraph}) is for us a $1$-dimensional simplicial complex (resp.~a $1$-dimensional cell complex). In other words, graphs are multigraphs that do not have loops or double edges. Given any simplicial complex, we call ``free'' any face that is properly contained in only one other face. The free faces in a graph are called {\mypink leaves}; some complexes have no free face. An  {\mypink elementary collapse} is the deletion of a single free face (and of the other face containing it).

 \begin{definition}[Extensively collapsible]
A complex $C$ is called {\mypink extensively-collapsible} if any sequence of elementary collapses reduces $C$ to a complex that is itself collapsible. In other words, $C$ is extensively collapsible if and only if by performing elementary collapses, we never get stuck. 
We also say that $C$ is {\mypink extensively-collapsible onto $D$} if any sequence of elementary collapses that does not delete faces of $D$, reduces $C$ to a complex that is itself  collapsible to $D$. 
\end{definition}

For example, trees are extensively collapsible; in fact, every tree is extensively collapsible onto any of its subtrees. It is well-known that all collapsible $2$-complexes are also extensively-collapsible, cf.~e.g.~\cite{HogAngeloni}. However, an $8$-vertex example of a collapsible but not extensively-collapsible complex (in fact, a $3$-ball) was given in \cite{BL-Dunce}.

\begin{lemma} \label{lem:LCorder}
Let $C$ be a cycle. Let $\mathfrak{M}$ be any planar matching, partial or complete, of the edges of $C$.
Let $G$ be the multigraph obtained from $C$ by pairwise identifying the edges according to $\mathfrak{M}$ (preserving orientation). The following are equivalent:
\begin{compactenum}[\rm (1)]
\item $G$ contains at most one cycle;
\item $G$ can be obtained from $C$ via some sequence of LC gluings.
\end{compactenum}
\end{lemma}

\begin{figure}[htbp]
\begin{center}
\includegraphics[scale=0.67]{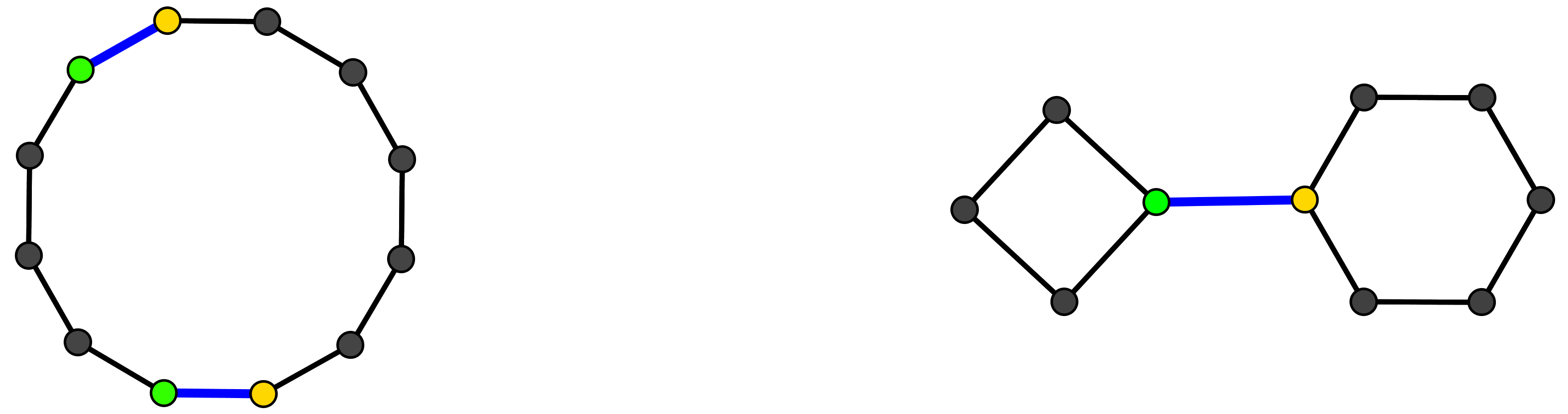}
\caption{A 12-gon (left) and the (multi)graph obtained from it  by identifying only the two blue edges (right). Note that this graph cannot be obtained from the $12$-gon via LC gluings, because the only two edges we are supposed to glue in the $12$-gon are far apart from one another.}
\label{fig:Matching}
\end{center}
\vskip-4mm
\end{figure}

\begin{proof}  
``(2) $\Rightarrow$ (1)'': Let $C$ be a cycle. Any LC gluing of two edges of $C$ either
\begin{compactitem}[ -- ]
\item preserves the number of cycles (if the edges share only one vertex), or
\item `kills' one cycle (in case the edges have both endpoints in common).
\end{compactitem}
So when we perform local gluings on a multigraph, the total number of cycles can only decrease. Since we started with a cycle, $G$ contains at most one cycle.

\smallskip \noindent 
 ``(1) $\Rightarrow$ (2)'':
 Fix a planar matching $\mathfrak{M}$  of $C$. If $\mathfrak{M}$  is a complete matching, the resulting multigraph $G$ will be a tree; if instead it is partial, $G$ will be a cycle with some trees attached. We proceed by induction on the number $n$ of edges of $C$, the case $3 \le n \le 6$ being easy. Let $e', e''$ be two edges of $C$ that are matched in $\mathfrak{M}$. If $e', e''$ are adjacent in $C$, their identification is an LC gluing, and there is nothing to show; so we shall assume they are not adjacent.
The effect of the gluing $e' \equiv e''$ is to squeeze $C$ into a left cycle $L$  and a right cycle~$R$, bridged by a single edge $e$ (as in Figure \ref{fig:Matching}). Moreover, $\mathfrak{M}$ restricts to planar matchings on both $L$ and $R$. Of these two ``submatchings'', at least one has to be complete, otherwise the final multigraph $G$ would contain at least two cycles. We will assume the submatching on $L$ is complete, the other case being symmetric.  Let $G_L$ be the subtree of $G$ corresponding to the edges of $L$. Let $v = e \cap G_L$.  
Choose a collapsing sequence of the tree $G_L$ onto $v$. This yields a natural ordering
$e_0 , \; e_1 , \; \ldots \; , \, e_{k-1}$
of the $k$ edges of the tree $G_L$, where $e_i$ is the  $i$-th edge to be collapsed and $e_{k-1}$ contains~$v$. 
Observe that $e_0$ must be a leaf of $G_L$; it corresponds therefore to a pair of adjacent edges $e'_0$ and $e''_0$ of $L$ matched under $\mathfrak{M}$. Recursively,  
for each $i$, the edges $e'_{i}$ and $e''_{i}$ become adjacent once we have identified $e'_{j}$ with $e''_{j}$, for all $j <i$.  In other words, the identifications $(e'_{i} \equiv  e''_{i})_{0 \le i \le k-1}$, performed in this order, are legitimate LC gluings. Now we are ready to rearrange the sequence, by postponing  the initial step $e' \equiv e''$. So let us set $e_k := e$, $e'_{k}:=e'$ and  $e''_{k}:=e''$. Starting from the initial cycle $C$, let us  perform $(e'_{k} \equiv  e''_{k})$ after all of the gluings $(e'_{i} \equiv  e''_{i})_{0 \le i \le k-1}$ have been carried out. 
The advantage is that $e'_k \equiv e''_k$ is now an LC step, because $e_{k-1}$ and $e_k$ both contained the vertex $v$ (so after $e'_{k-1}$ and $e''_{k-1}$ are identified, the edges $e'_k \equiv e''_k$ become incident at $v$). 

We are eventually left with the right cycle $R$. The subgraph $G_R$ of $G$ corresponding to the edges of $R$ contains at most one cycle. By inductive assumption, $G_R$ can be obtained from $R$ via a sequence of LC gluings. The latter sequence, performed after  $(e'_{i} \equiv  e''_{i})_{0 \le i \le k}$, forms a longer sequence of LC gluings that constructs $G$ from $C$.
\end{proof}

\begin{remark}
Topologically, the proof above can be recapped  as follows. Initially, we have an ``unwanted'' non-LC gluing $e' \equiv e''$ that increases the number of cycles from $1$ to $2$. Since in the end the graph $G$ produced has at most $1$ cycle, at some point the extra cycle has to be suppressed. The only way to suppress a cycle with a planar matching, is to identify some pair of edges $f'$, $f''$ that have both endpoints in common. Our proof strategy was:
\begin{compactitem}
\item to postpone the gluing $e' \equiv e''$, so that it is becomes an LC gluing; and also
\item to anticipate $f' \equiv f''$, so that these two edges are glued when they only share \emph{one} of their endpoints, not both. 
\end{compactitem}
We did not change the matching; we only changed the order in which the matching is performed. But in the rearranged sequence, no step increases the number of cycles by one. (There is also one less step that \emph{decreases} the number of cycles by one; these two steps `canceled out'.)
\end{remark}

Here is a variation we will need in the next Section. Given a graph $G$, we say that a vertex $v$ of $G$ is {\mypink active} if it belongs to a cycle. For example, every vertex of a cycle $C$ is active. If we perform an LC gluing of two adjacent edges of $C$, the vertex between the two edges gets ``de-activated''. In a tree, no vertex is active.

\begin{lemma} \label{lem:LCorder2}
Let $C$ be a cycle. Let $\mathfrak{M}$ be any \emph{complete} planar matching of the edges of $C$. Let $G$ be the tree obtained from $C$ by pairwise identifying the edges according to $\mathfrak{M}$, as in the previous Lemma. 
Given an arbitrary vertex $c_0$ of $C$, there is a sequence of LC gluings that produced $G$ from $C$ and in which the vertex $c_0$ is active until the very last gluing.
\end{lemma}

\begin{proof}
Since every tree is simplicially collapsible onto any of its vertices, we may choose a collapsing sequence of $G$ onto the vertex corresponding to $c_0$. 
Now, every pair of adjacent edges in $C$ matched by $\mathfrak{M}$ corresponds to a leaf in the tree $G$; and elementary collapses in $G$ (which are just leaf deletions) correspond to LC gluings on $C$. Hence, our collapse of $G$ onto $c_0$ induces a sequence of LC gluings, the last of which identifies two edges sharing both endpoints (one of the endpoints  being $c_0$). 
\end{proof}

\begin{remark} Unlike Lemma~\ref{lem:LCorder}, Lemma~\ref{lem:LCorder2} does not extend to partial matchings. For example, let us start with a hexagon of vertices $\{a,b,c,d,e,f\}$, and let us identify $[b,c]$ and $[c,d]$ (preserving orientation). This makes $b$ coincide with $d$. Let us then glue together the edges $[a,b]$ and $[d,e]$, which have just become adjacent. The resulting partial matching $\mathfrak{M}$ satisfies the condition of Lemma~\ref{lem:LCorder}; however, there is only one possible possible sequence of LC gluings realizing $\mathfrak{M}$, and this only possible sequence deactivates the vertex $c$ in the \emph{first} step. 
\end{remark}

\begin{definition}[Extensively LC]
Let $P$ be a $d$-dimensional pseudomanifold.  
We say that $P$ is {\mypink extensively LC} if, for \emph{any} spanning tree $T$ of the dual graph of $B$, (a complex combinatorially equivalent to) $P$ can be obtained via LC gluings from the tree of $d$-simplices $T_N$ dual to $T$.
\end{definition}

If we replace ``any'' with ``some'' in the definition above, we recover the classical definition of LC. Hence, ``extensively-LC'' trivially implies LC. See Remark \ref{rem:difference} below for the difference.

\begin{proposition}[{essentially Durhuus \cite{Durhuus}}]
All $2$-balls and $2$-spheres are extensively LC. 
\end{proposition}

\begin{proof} Let $B$ be an arbitrary $2$-sphere or $2$-ball. Let $T$ and $T_N$ be as in the definition of extensively-LC. By construction, we know that $B$ is obtained from $T_N$ by some matching $\mathfrak{M}$ of the edges of $\partial T_N$, which is a $1$-dimensional sphere (or in other words, a cycle). Note that the matching is uniquely determined once the tree $T_N$ is chosen. If $B$ is a $2$-sphere, the matching is complete; if $B$ is a ball, $\partial B$ is a cycle, the matching is partial, and the edges left unmatched are precisely the edges of $\partial B$. In both cases, the multigraph obtained from $\partial T_N$ via the identifications in $\mathfrak{M}$ contains at most one cycle. Using Lemma \ref{lem:LCorder}, we conclude.
\end{proof}

If $T$ is a spanning tree of the dual graph of a (connected) $d$-manifold, following  \cite[p.~214]{BZ} we denote by {\mypink $K^T$} the {$(d-1)$}-di\-mensional subcomplex of the manifold determined by all the $(d-1)$-faces that are not intersected by $T$. 
When $d=3$, $K^T$ is $2$-dimensional. Recall that for $2$-complexes collapsibility and extensive-collapsibility are equivalent notions. Using this, it is an easy exercise to adapt the original proofs of \cite[Corollary 2.11]{BZ} and of \cite[Corollary 3.11]{BZ}, respectively, to derive the following results:

 \begin{theorem} \label{thm:NonELC}
 Let $S$ be a triangulated $3$-sphere. The following are equivalent:
 \begin{compactenum}[\rm (i)]
 \item $S$ is extensively-LC;
 \item for every spanning tree $T$ of the dual graph of $S$, the complex $K^T$ is collapsible;
 \item for every tetrahedron $\Delta$ of $S$, the $3$-ball $S - \Delta$ is extensively collapsible;
 \item for some tetrahedron $\Delta$ of $S$, the $3$-ball $S - \Delta$ is extensively collapsible.
 \end{compactenum}
\end{theorem}

 \begin{theorem}  \label{thm:NonELC2} Let $B$ be a triangulated $3$-ball. The following are equivalent:
 \begin{compactenum}[\rm (i)]
 \item $B$ is extensively-LC;
 \item for some tetrahedron $\Delta$, the $3$-ball $B - \Delta$ is extensively collapsible to $\partial B$;
 \item for every tetrahedron $\Delta$, the $3$-ball $B - \Delta$ is extensively collapsible to $\partial B$;
\item for every spanning tree $T$ of the dual graph of $S$, the complex $K^T$ collapses to $\partial B$.
 \end{compactenum}
\end{theorem}

\begin{corollary}
Every triangulated $d$-ball or $d$-sphere with less than $8$ vertices is extensively-LC.
\end{corollary}

\begin{proof}
By a result of Bagchi and Datta \cite{BagchiDatta}, all acyclic $2$-complexes with less than $8$ vertices are collapsible; it follows that all collapsible $2$-complexes with less than $8$ vertices are extensively collapsible \cite{BL-Dunce}.
\end{proof}

\begin{remark} \label{rem:difference}
Some $3$-sphere with $8$ vertices that is LC, but not extensively, is presented in~\cite{BL-Dunce}. After we remove a tetrahedron from such sphere, we obtain a collapsible ball $B$; but there is also a sequence of elementary collapses that from $B$ gets us stuck in an $8$-vertex triangulation of the Dunce Hat~\cite{BL-Dunce}. (See also {\rm \cite[pp.~107--109]{BThesis}} for a similar example with $12$ vertices.) Moreover, the boundary of the $7$-simplex is not extensively-LC, since (after the removal of an arbitrary $6$-face) there is a sequence that gets us stuck in a $8$-vertex Dunce Hat: This was first shown by Crowley et al., cf.~\cite{Cetal} \cite[Section~5.3]{BL-Exp}.
\end{remark}

\section{The only Mogami nucleus is the simplex}
Let us  now focus on $d=3$. We wish to study how  LC or Mogami steps in a construction of a $3$-manifold affect the boundary-link of a single vertex. The four examples we present will be crucial in the proof of our Main Theorems. First we need one additional notation.

\begin{definition}[Merging]
Let $C, D$ be two cycles with an edge $e$ in common. The {\mypink merging} operation produces a new cycle as follows: we take the union $C \cup D$, and we delete the edge $e$. 
\end{definition}

\begin{example} \label{ex:HomA}  \em
Let $B$ be a $3$-ball. Let $e$ be an edge in $\partial B$. Let $v$ and $w$ be the two vertices in $\link(e, \partial B)$. If we identify the two triangles $v \ast e$ and $w \ast e$, this is a legitimate LC gluing -- in fact, a \fold. Let $Q$ be the pseudomanifold obtained. Topologically, $Q$ is also a $3$-ball. With slight abuse of notation, let us keep calling $v$ be the vertex of $Q$ resulting from the identification of $v$ and $w$. It is easy to see that  $\link(v,\partial Q)$ is the cycle obtained by merging $\link(v,\partial B)$ and $\link(w, \partial B)$.
\end{example}

\begin{example} \label{ex:HomB} \em
Let $B$ be a $3$-ball. Let $x$ be a vertex in  $\partial B$. Let $e_1$, $e_2$ be two edges in $\link(x, \partial B)$. If we identify the two triangles $x \ast e_1$ and $x \ast e_2$, this is a legitimate Mogami gluing. Let $v_1, w_1$ be the two endpoints of $e_1$. Similarly, let $v_2, w_2$ be the two endpoint of $e_2$, labeled so that the vertex that is identified to $v_1$ is $v_2$. Let $Q$ be the obtained pseudomanifold (which is not a ball, this time.) Let us call $v$ the vertex of $Q$ resulting from the identification of the two vertices $v_1$ and $v_2$. It is easy to see that $\link(v,\partial Q)$ is a cycle. It is obtained from $C_1=\link(v_1,\partial B)$ and $C_2=\link(v_2, \partial B)$ with an operation that is an LC gluing plus a merging. More precisely, $C_1$ and $C_2$ do not have an edge in common; they share only the vertex $x$. However, the cycle $\link(v,\partial Q)$ can be obtained from $C_1$ and $C_2$ by first identifying $[x,w_1]$ (which is in $C_1$) and $[x,w_2]$ (which is in $C_2$), and then by performing a merging at the resulting edge $[x,w]$. 
\end{example}

\begin{example} \label{ex:NonHom} \em
Let $P$ be a pseudomanifold obtained from a $3$-ball by performing one Mogami gluing of $2$ triangles sharing only a vertex $v$, and then another Mogami gluing of $2$ triangles sharing only a vertex $w \ne v$, such that $v$ and $w$ belong to adjacent triangles in $P$. Then: 
\begin{compactitem}
\item $\link(v, \partial P)$ is the disjoint union of two cycles, $A_v$ and $B_v$; 
\item $\link(w, \partial P)$ is also the disjoint union of two cycles, $A_w$ and $B_w$;
\item $\link(v, \partial P) \cap \link(w, \partial P)$ consists of an edge $e$, which (up to relabeling) belongs to $A_v \cap A_w$. 
\end{compactitem}
Let us identify the two triangles $v \ast e$ and $w\ast e$, and let $Q$ be the resulting pseudomanifold. With the usual abuse of notation, let us call $v$ be the vertex of $Q$ obtained from the identification of $v$ and $w$. It is easy to see that $\link(v,\partial Q)$ is a disjoint union of \emph{three} cycles, namely $B_v$, $B_w$, and a third cycle obtained by merging $A_v$ and $A_w$. In particular, $\partial Q$ is not homeomorphic to $\partial P$. (This pathology is due to the presence of two different singularities in $P$, which are identified in the gluing; on LC pseudomanifolds, 
{\fold} does preserve the homeomorphism type).
\end{example}

\begin{example} \label{ex:Impossible} \em
Let us start with an annulus of $4$ squares, and let us subdivide each square into four triangles by inserting the two diagonals (Figure \ref{fig:Annulus}). Let $w$ be one of the four square barycenters. Let $a,b,c,d$ be the four corners of the square containing $w$, labeled  so that $ab$ and $cd$ are free edges (i.e.~edges that belong to one triangle only). Let $A$ be the obtained simplicial complex. Let us take a cone $v \ast A$ from a vertex $v$ outside $A$. This $v \ast A$ is a Mogami pseudomanifold by Proposition \ref{prp:Cone}. 
The boundary-link of $v$ consists of $2$ squares. Note that from $v \ast A$ one can easily obtain a Mogami $3$-ball without interior vertices by pairwise identifying the top four triangles.
Instead, from $v \ast A$, let us perform a {\fold} step by gluing $[c,d,v]$ with $[c,d,w]$. Let $P$ be the obtained pseudomanifold. Since $v \ast A$ contained triangles $[a,b,v]$ and $[a,b,w]$, now that $v$ is carried onto $w$ we have in $P$ two distinct triangles $\Delta_1$ and $\Delta_2$ that share one edge and also the opposite vertex. Hence $P$ (which topologically is homeomorphic to $v \ast A$, cf.~Example~\ref{ex:HomA}) is not a simplicial complex.\end{example}

\begin{figure}[htbp]
\begin{center}
\includegraphics[scale=0.7]{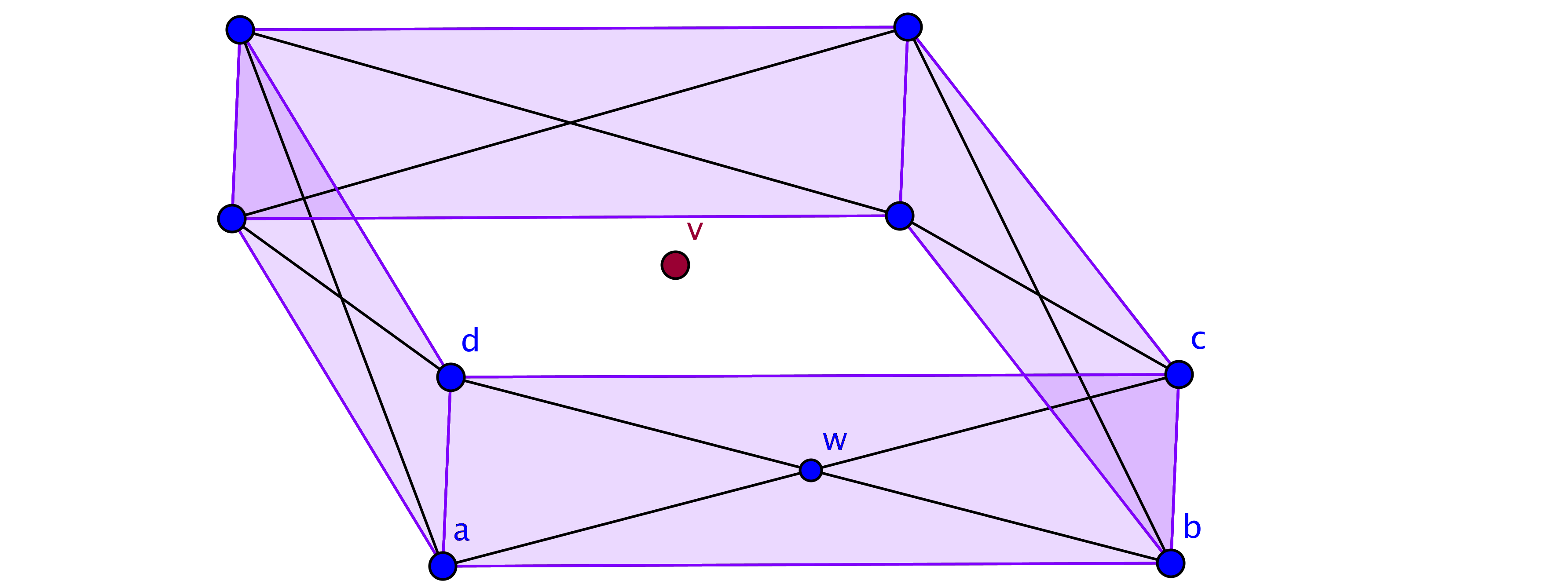}
\caption{The triangulated annulus $A$ (in purple). The cone $v \ast A$ is a Mogami pseudomanifold, with a singularity at $v$.  With a couple of LC gluings (on the four top triangles), from $v \ast A$ it is possible to get a $3$-ball without interior vertices. But if from $v \ast A$ we glue together the triangles $[v,c,d]$ and $[c,d,w]$, this is a faux pas (Remark \ref{rem:Impossible}).}
\label{fig:Annulus}
\end{center}
\vskip-3mm
\end{figure}

\begin{remark} \label{rem:Impossible} No $3$-ball without interior vertices can be obtained via Mogami gluings from the pseudomanifold $P$ of Example \ref{ex:Impossible}.
In fact, assume by contradiction that $R$ is a ball without interior vertices obtained with a Mogami construction from $P$. 
Because $R$ is a simplicial complex, in $R$ the two triangles $\Delta_1$ and $\Delta_2$  that share the vertices $v,a,b$ must be identified at some point; so we might as well glue them immediately. Let us call  $Q$ be the pseudomanifold obtained from $P$ by gluing $\Delta_1 \equiv \Delta_2$.  We may assume that $R$ is obtained via Mogami gluings from $Q$. 
Since $R$ is a ball without interior vertices, the link of $v$ in $R$ must be a disk. Since $Q$ is not a simplicial complex the notion of ``link of $v$ in $Q$'' is not well-defined; but we can look at the {\em spherical link} $L$ of $v$ in $Q$, which is what we would get by intersecting $Q$ with a sphere of small radius centered at $v$. (In simplicial complexes, this is isomorphic to the vertex link.) Up to homeomorphism, we can think of $L$ as a $2$-dimensional simplicial complex obtained from a finely triangulated annulus by identifying (coherently, without twists) two parallel edges in {\em different} components of the boundary. Note that $L$ is not planar, in the sense that no simplicial complex homeomorphic to $L$ can be drawn in $\mathbb{R}^2$ without self-intersections. Now, any further Mogami step performed on $Q$ will possibly modify $L$ only via identifications in its boundary. 
Topologically, these steps may transform the spherical link of $v$ into a torus, but not into a $2$-ball (or a $2$-sphere). A contradiction.
\end{remark}

In fact, the topological argument of Remark \ref{rem:Impossible} above proves the following:

\begin{lemma}\label{lem:Impossible} 
Let $M$ be a Mogami pseudomanifold obtained from a $3$-ball via a single Mogami gluing that is not an LC gluing. Let $v$ be the singular vertex of $M$. Let $C_1$ and $C_2$ be the two disjoint components of the boundary-link of $v$. 
Suppose there is a vertex $w \ne v$ in $\partial M$ such that $\link (v,\partial M) \cap \link (w,\partial M)$ consists of $2$ edges, one in $C_1$ (say, $[a,b]$) and one in $C_2$ (say, $[c,d]$). Let $P$ be the pseudomanifold obtained from $M$ via the LC gluing that identifies the triangles $[v,c,d]$ and $[c,d,w]$.  There is no $3$-ball without interior vertices that can be obtained via Mogami gluings from the pseudomanifold $P$.
\end{lemma}

We are now ready  to prove our main result.

\begin{theorem}\label{thm:main}
Let $B$ be a Mogami $3$-ball without interior vertices. Let $T_N$ be the tree of tetrahedra from which $B$ is constructed, via some sequence of Mogami gluings. Then, $B$ can also be constructed from $T_N$ via some sequence of {\em LC} gluings. In particular, all Mogami $3$-ball without interior vertices are LC.
\end{theorem}

\begin{proof}
If all Mogami gluings are LC gluings, there is nothing to prove. Otherwise, let us consider the first Mogami gluing $\Delta'_0 \equiv \Delta''_0$ that is not LC. Let $v =  \Delta'_0 \cap \Delta''_0$. By definition there are disjoint edges $\delta'_0, \delta''_0$ such that $\Delta'_0 = v \ast \delta'_0$ and $\Delta''_0 = v \ast \delta''_0$. Let $P$ be the pseudomanifold obtained after the gluing $\Delta'_0 \equiv \Delta''_0$; the vertex $v$ is in the boundary of $P$, while the triangle $\Delta_0$ obtained from the identification is in the interior of $P$. We denote by $\delta_0$ the edge opposite to $v$ in $\Delta_0$.
As we saw in Figure \ref{fig:VertexStar}, the gluing creates a singularity at $v$: namely, $\link(v, \partial P)$ consists of \emph{two} cycles. Since $B$ is a $3$-ball with all vertices on the boundary, the subsequent Mogami gluings in the construction of $B$  from $P$ will 
\begin{compactitem}
\item keep the vertex $v$ in the boundary, and
\item eventually ``kill'' one of the two connected components of $\link(v, \partial P)$. 
\end{compactitem}
Let us call $C$ the ``doomed'' component, that is, the cycle of  $\link(v, \partial P)$ none of whose edges will eventually appear in $\link(v, \partial B)$. Let us denote by $c_0$ the vertex of $\delta_0$ that belongs to $C$.

Our strategy is to consider this cycle $C$ and rearrange the sequence of gluings according to Lemma \ref{lem:LCorder2}, so that after the rearrangement, all gluings in the sequence are LC gluings, and the last pair of edges glued is a pair adjacent to the edge $\delta_0$. Before doing this, though, we need a delicate preliminary argument. In fact, while constructing $B$ from $P$, all triangles of $\star(v, \partial P)$ are going to be matched and sunk into the interior; but what we do not know for sure, is whether they are going to be matched to one another.  A priori, there are other two possibilities that we should consider (both of which could occur multiple times):
\begin{compactenum}[(a)]
\item for some edge $e$ of $C$, it could happen that $v \ast e$ is matched in an LC gluing with some triangle $w \ast e$ outside $\star(v,\partial P)$;
\item or it could also happen that $v \ast e$ is matched in a Mogami gluing to another triangle that does not contain $v$, but contains exactly one of the two endpoints (let us call it $x$) of $e$. 
\end{compactenum}
The steps above affect the boundary-link of $v$ as follows.
\begin{compactenum}[I)]
\item The cycle $C$ is ``expanded'' via a merging operation. For example, in case (a) the boundary-link of $v$ gets merged with the boundary-link of $w$, as explained in  Example \ref{ex:HomA}. Case (b) is similar: The vertex $v$ is identified with a vertex $v_2$ of the other triangle, and essentially the boundary-link of $v$ gets merged with the boundary-link of $v_2$ (after an LC gluing; compare Example \ref{ex:HomB}.)
\item Possibly, the link of $P$ might acquire further connected components. This happens when the vertex $w$ identified with $v$ is also a singularity, a case we saw  in Example \ref{ex:NonHom}. 
\end{compactenum}
These cases, however, do not ruin our proof strategy --- they just delay it. Our remedy in fact is to anticipate all matchings of the type (a) and (b) described above, in a ``first round'' of identifications. For example, if a single triangle $v \ast e$ is later matched in a {\fold} with some triangle $w \ast e$ outside $\star(v,\partial P)$, then we can rearrange the sequence by performing such LC gluing immediately. After all identifications of type (a) and (b) have been carried out, if  $P_1$ is the resulting pseudomanifold, we ask ourselves again: are all triangles of $\star(v,\partial P_1)$ going to be matched exclusively with one another? If not, we repeat the procedure above, in a second round of identifications, and we call the obtained pseudomanifold $P_2$. And so on. 

The effect of these rounds on the boundary-link of $v$ is to expand it by inglobating new edges. We make a crucial claim: in these rounds of identifications, the components of the boundary-link of $v$ remain separate. The proof of this claim relies on Lemma \ref{lem:Impossible}. In fact, suppose by contradiction that passing from $P_1$ to $P_2$, say, we have included into $C$ an edge $[a,b]$ that belongs to another component of the boundary-link of $v$ (which is what we have done in Example \ref{ex:Impossible}.) This means that in $P_2$ we have a singularity $v$, and two distinct triangles containing $v$ and the edge $[a,b]$. So if we want to obtain a simplicial complex, we are forced to glue the two triangles together; and with the same proof of Remark \ref{rem:Impossible}, no matter how we continue this Mogami construction, we are never going to achieve a $3$-ball without interior vertices. A contradiction. (This shows that the components of the boundary-link of $v$ never have an \emph{edge} in common; in analogous way, adapting Lemma \ref{lem:Impossible}, one proves they cannot have \emph{vertices} in common, either.) 

Eventually, after a finite number of rounds, we will reach a pseudomanifold $P'$ such that:
\begin{compactitem}
\item $\link(v,\partial P')$ consists of $k \ge 2$ connected components, 
\item $B$ is obtained with a list of Mogami gluings from $P'$, a process in which exactly $k-1$ of the components of $\link(v,\partial P')$ are going to be ``killed'',
\item if $C'$ is the connected component of $\link(v,\partial P')$ obtained from $C$ via merging operations, then for any edge $e$ of $C'$ there exists an edge $f$ of $C'$ such that, in one of the Mogami gluings that leads from $P'$ to $B$, the triangle $v \ast e$ is identified with $v \ast f$ .
\end{compactitem}
In fact, we can repeat the reasoning above until the last property holds for {\em all} the $k-1$ ``doomed'' connected components  of $\link(v,\partial P')$.

Note that $C'$ contains \emph{all} vertices of $C$. This is because the merging operation does not delete any vertex. In particular, the vertex $c_0= \delta_0 \cap C$ of $C$ will be present in $C'$ as well.

We are now in the position to use Lemma \ref{lem:LCorder2}. The Mogami construction that leads from $P'$ to $B$ yields a complete matching of the edges of $C'$. Clearly, ordering the edges in  $\link (v, \partial P')$ is the same as ordering the triangles in $\star (v, \partial P')$; also, two edges $e$, $f$ are adjacent in the link of $v$ if and only if $v \ast e$ and $v \ast f$ are adjacent in the star of $v$. Let us thus reorder the gluings involving triangles in $\star (v, \partial P')$, according to Lemma \ref{lem:LCorder2}, so that the vertex $c_0$ is deactivated last.   In this order, the identifications ``killing'' the component $C'$ are all LC gluings. Furthermore, it is easy to see that \emph{all} gluings mentioned above (those leading from $P$ to $P'$, plus all LC gluings that kill $C'$) can be performed before the identification $\Delta'_0 \equiv \Delta''_0$. With this postponement the step $\Delta'_0 \equiv \Delta''_0$ becomes an LC gluing: In fact, after all other identifications have been carried out, $\Delta'_0$ and $\Delta''_0$ share the edge $[v,c_0]$.
In conclusion, by reshuffling the Mogami sequence we got rid of the first non-LC step. By induction, we reach our claim. 
\end{proof}

\begin{corollary} \label{cor:main} Let $B$ be a $3$-ball without interior vertices. The following are equivalent:
\begin{compactenum}[\rm (1)]
\item $B$ is Mogami;
\item $B$ is LC;
\item some (possibly empty) sequence of {\spread} operations reduces $B$ to a tree of tetrahedra;
\item $B$ has trivial nuclei (that is, some sequence of {\spread} and {\split} operations reduces $B$ to disjoint tetrahedra.)
\end{compactenum}
\end{corollary}

\begin{proof}
``(1) $\Leftrightarrow$ (2)'' follows from Theorem \ref{thm:main}. 

``(2) $\Leftrightarrow$ (3)'': In \cite[Lemma 3.18]{BZ} it is shown that a $3$-ball $B$ without interior vertices is LC if and only if it can be obtained from a tree of tetrahedra via {\fold} steps. The conclusion follows by reversing the construction. 

``(3) $\Leftrightarrow$ (4)'': It follows from the characterization of tree of $N$ tetrahedra as the complexes that can be reduced to $N$ disjoint tetrahedra using $N-1$ {\split} operations. 
\end{proof}

\begin{corollary}\label{cor:Nuclei}
The only Mogami nucleus is the tetrahedron.
\end{corollary}

\begin{proof} Clearly a tetrahedron is Mogami. On any other nucleus, neither {\spread} nor {\split} steps are possible, because every interior triangle has at most one edge on the boundary. 
\end{proof}

\begin{corollary}
Some $3$-balls are not Mogami. 
\end{corollary}

For example, Hachimori's triangulation of Bing's (thickened) house with two rooms, described in \cite{HachiLibrary} \cite[p.~89]{HachiThesis}, is a non-trivial nucleus with $1554$ tetrahedra.  
The smallest non-trivial nucleus found so far with computer tools has only $37$ tetrahedra \cite[p.~260]{CEY}. 

\subsection*{Relation with knots and collapsibility}
A {\mypink spanning edge} in a $3$-ball $B$ is an interior edge with both endpoints on the boundary. A spanning edge $[x,y]$ is called {\mypink knotted} if some (or equivalently, any) path in $\partial B$ from $x$ to $y$, together with the edge $[x,y]$, forms a non-trivial knot.   For brevity, we call a $3$-ball $B$ {\mypink knotted } if it contains a knotted spanning edge. 

Using the same exact proof of \cite[Proposition 3.19]{BZ}, one can obtain the following consequence of Corollary \ref{cor:main}:

\begin{lemma} \label{lem:Knot1}
Mogami $3$-balls without interior vertices do not contain knotted spanning edges. 
\end{lemma}

Compare the result above with the following Lemma (which is known,  but we include a proof for completeness):

\begin{lemma} \label{lem:Knot2}
Every (tame, non-trivial) knot can be realized as knotted spanning edge in some $3$-ball without interior vertices.
\end{lemma}

\begin{proof}
The following classical construction goes back to Furch \cite{Furch}: Let us dig a hole, shaped like the chosen knot, inside a suitably fine pile of cubes, stopping one step before destroying the property of having a $3$-ball. Let us then triangulate every cube according to a standard pulling triangulation. This construction is carried out in detail for the trefoil knot by Hachimori \cite[model ``Furch's knotted hole ball'']{HachiLibrary}; compare also \cite[Example 2.14 \& Figure 3]{BZ}. Observe that the $3$-ball obtained with such construction typically contains plenty of interior vertices. However, we can progressively ``shell out'' all cubes that are far away from the knot, until we reach a thinner triangulation without interior vertices. (Another way to reach such triangulation is to simply apply the spread and the split operations greedily. It is an easy topological exercise, essentially analogous to the proof of Lemma \ref{lem:Knot1}, to check that none of these operations can delete or modify the existing knot.)
\end{proof}

From Lemma \ref{lem:Knot1} and  Lemma \ref{lem:Knot2}, we can find \emph{infinitely many} examples of $3$-balls that are not Mogami. 
In fact, we can prove an asymptotic enumeration result:

\begin{proposition} \label{prop:asymptotics}
In terms of the number $N$ of facets, the number of non-Mogami $3$-balls without interior vertices is asymptotically the same of the total number of $3$-balls without interior vertices. 
\end{proposition}

\begin{proof}
Let us fix a $3$-ball  $A$ with  some knotted spanning edge and with all vertices on $\partial A$ (cf.~Lemma \ref{lem:Knot2}). Let $F_A$ be the number of facets of $A$. Let us also fix a triangle $\Delta_A \subset \partial A$. Now, let $B$ be an arbitrary $3$-ball with $N $ tetrahedra, without interior vertices, and with a distinguished triangle $\Delta_B \subset \partial B$. From $B$ we can obtain a $3$-ball $B'$ with $N+F_A$ tetrahedra via a {\unite} step that consolidates the $3$-balls $A$ and $B$ by identifying $\Delta_A \equiv \Delta_B$. (Ignore the fact that there are multiple ways to do this, according to rotation, as this amounts to an asymptotically neglectable factor.) No matter how we choose $B$, the union $B'=A \cup B$ is going to contain the same knotted spanning edge of $A$. But  since all its vertices are on the boundary, by Lemma  \ref{lem:Knot1} the ball $B'$ cannot be Mogami. Now note that $B'$ determines $B$: In fact, for any interior triangle $\Delta$ of $B'$ with all three edges on $\partial B'$, we could split $B'$ at $\Delta$ and check if one of the two $3$-balls obtained is combinatorially equivalent to $A$ (if it is, the other $3$-ball is $B$). Hence the transition from $B$ to $B'$ yields an injective map 
\[ 
\left\{ \begin{array}{l}
\textrm{$3$-balls with $N$ tetrahedra } \\
\textrm{and with $0$ interior vertices}  
\end{array} \right \} \: \:
\longhookrightarrow \: \:
\left\{ \begin{array}{l}
\textrm{non-Mogami $3$-balls with $N+F_A$ tetrahedra} \\
\textrm{and with $0$ interior vertices}  
\end{array} \right \} \]
If we pass to the cardinalities and let $N$ tend to infinity, $F_A$ being constant, we conclude.
\end{proof}

Finally, we recall the connection of knot theory with simplicial collapsibility:

\begin{proposition}[{essentially Goodrick, cf. \cite[Corollary 4.25]{Benedetti-DMT4MWB}}] \label{prop:2gen}
Let $K$ be any knot whose group admits no presentation with $2$ generators. (For example, the double trefoil). Any knot with a knotted spanning edge isotopic to $K$, cannot be collapsible.
\end{proposition}

\begin{proposition}[{\cite[Theorem 3.23]{BZ}}] \label{prop:2bridge}
For any $2$-bridge knot $K$ (for example, the trefoil), there is a collapsible $3$-ball without interior vertices with a knotted spanning edge isotopic to $K$.
\end{proposition}

Summing up, we have the following hierarchy:

\begin{theorem} \label{thm:final}
For $3$-balls \emph{without interior vertices}, the following inclusions hold: \em
\[
\{  \textrm{shellable}  \} \subsetneq 
\{  \textrm{LC}  \} =
\{  \textrm{Mogami}  \} \subsetneq
\{  \textrm{collapsible}  \} \subsetneq 
\{  \textrm{all $3$-balls without interior vertices}  \}   .
\]
\end{theorem}

\begin{proof}
Any linear subdivision of a (convex) $3$-dimensional polytope (with or without interior vertices) is collapsible \cite{Chillingworth} and even LC \cite[Theorem 3.27]{BZ}. However, Rudin proved in 1958 that not all these linear subdivisions are shellable \cite{Rudin}; her counterexample, known as ``Rudin's ball'', is a subdivision of a tetrahedron with all $14$ vertices on the boundary. The equivalence of LC and Mogami is discussed in Corollary \ref{cor:main}. Any knotted $3$-ball described in Proposition \ref{prop:2bridge} is collapsible, but cannot be Mogami by Lemma \ref{lem:Knot1}. Finally, $3$-balls without interior vertices that are not collapsible can be produced by pairing together Lemma \ref{lem:Knot2} and Proposition \ref{prop:2gen}: For example, any $3$-ball without interior vertices and with a knotted spanning edge isotopic to the double trefoil would do.
\end{proof}

\section*{Acknowledgments}
The author wishes to thank G\"unter M.~Ziegler for improving the introduction and Jean--Pierre Eckmann for useful discussions.



\end{document}